\documentclass[12pt,a4paper]{article}
\title{Efficiency in Truthful Auctions\\ via a Social Network}
\author{Seiji Takanashi\footnote{Graduate School of Economics, Kyoto University, Yoshidahon-machi, Sakyo-ku, Kyoto-shi 606-8501, Japan; E-mail: s.takanashi1990@gmail.com}, Takehiro Kawasaki\footnote{Graduate School of Information Science and Electrical Engineering, Kyushu University, 744, Motooka, Nishi-ku, Fukuoka-shi 819-0395, Japan; E-mail: kawasaki@agent.inf.kyushu-u.ac.jp}, Taiki Todo\footnote{Graduate School of Information Science and Electrical Engineering, Kyushu University, 744, Motooka, Nishi-ku, Fukuoka-shi 819-0395, Japan; E-mail: todo@inf.kyushu-u.ac.jp}, Makoto Yokoo\footnote{Graduate School of Information Science and Electrical Engineering, Kyushu University, 744, Motooka, Nishi-ku, Fukuoka-shi 819-0395, Japan; E-mail: yokoo@inf.kyushu-u.ac.jp}}
\date{\today}
\usepackage{amsmath,amssymb}
\usepackage{amsfonts}
\usepackage{enumerate}
\usepackage{amsthm}
\usepackage{cases}
\usepackage{color}
\usepackage[pdftex]{graphicx}
\usepackage[authoryear]{natbib}

\begin{document}
\bibliographystyle{ecta}
\interfootnotelinepenalty=10000
\theoremstyle{definition}
\newtheorem{theo}{Theorem}
\newtheorem{defi}{Definition}
\newtheorem{lemm}{Lemma}
\newtheorem{prop}{Proposition}
\newtheorem{coro}{Corollary}

\maketitle 

\newcommand{\argmax}{\mathop{\rm arg~max}\limits}
\newcommand{\argmin}{\mathop{\rm arg~min}\limits}
\maketitle

\begin{abstract}
In this paper, we study efficiency in truthful auctions via a social network, where a seller can only spread the information of an auction to the buyers through the buyers' network. In single-item auctions, we show that no mechanism is strategy-proof, individually rational, efficient, and weakly budget balanced. In addition, we propose $\alpha$-APG mechanisms, a class of mechanisms which operate a trade-off between efficiency and weakly budget balancedness. In multi-item auctions, there already exists a strategy-proof mechanism when all buyers need only one item. However, we indicate a counter-example to strategy-proofness in this mechanism, and to the best of our knowledge, the question of finding a strategy-proof mechanism remains open. We assume that all buyers have decreasing marginal utility and propose a generalized APG mechanism that is strategy-proof and individually rational but not efficient. Importantly, we show that this mechanism achieves the largest efficiency measure among all strategy-proof mechanisms. 
\end{abstract}

\section{Introduction} 

Imagine that an auction will take place. Until it starts, the auctioneer wants to spread information about it to all the people who might have interest in the goods being sold. Firms might pay a lot for advertising before auctions occur. Many social network services (SNSs) allow connections among people who can spread advertisements. Many researchers have investigated advertisements and such social networks, e.g., \cite{emek2011mechanisms}, \cite{borgatti2009network}, \cite{jackson2010social}, and \cite{kempe2003maximizing}. In this situation, the auctioneer may assume that those who read the advertisements will spread such information through their SNSs. However, they may be reluctant because they themselves can benefit from fewer participants at the auction. If the auctioneer wants the participants to spread the word about the auction, he has to incentivize them to distribute it to their neighbors. Since standard auction theory has not studied such auctions with a network of buyers, results are scant about auction advertisements through a buyers' network. 

To the best of our knowledge, \cite{li2017mechanism} is the first paper that deals with this issue. They generalize a classical single-item auction to one with a social network, which consist of all buyers and a seller who is located somewhere in the network but is restricted from directly spreading the information to all the buyers. The authors assume that all the buyers only get information through the network and analyzed two well-known mechanisms, the Vickrey-Clarke-Groves (VCG) mechanism (\cite{vickrey1961counterspeculation}, \cite{clarke1971multipart}, \cite{groves1973incentives}) and the Information Diffusion Mechanism. They extended these two mechanisms to auctions via a social network and prove that they are strategy-proof, where no buyer has an incentive to untruthfully report their valuation OR to prevent the information from being spread to their neighbors. They also prove that the VCG mechanism is efficient but not weakly budget balanced, that is, the seller's revenue can be negative, and the Information Diffusion Mechanism is weakly budget balanced but not efficient. \cite{zhao2018selling} also study such auctions, generalize the Information Diffusion Mechanism proposed by \cite{li2017mechanism}, and show that the Generalized Information Diffusion Mechanism is strategy-proof, individually rational, and weakly budget balanced. However, in this paper we show a counter-example to strategy-proofness in this mechanism. 

The following are our results. Our model of auctions via a social network is almost the same as the one analyzed in these above two papers. In single-item auctions, we prove that no mechanism is strategy-proof, individually rational, weakly budget balanced, and efficient. In other words, under classical criteria, the seller's revenue might be negative if she sells the item to the buyer with the largest valuation. We propose for a single-item auction $\alpha$-APG mechanisms, which is a class of mechanisms which operate a trade-off between efficiency and weakly budget balancedness. In multi-item auctions, we indicate a counter-example to strategy-proofness in the Generalized Information Diffusion Mechanism proposed by \cite{zhao2018selling}. In addition, we analyze efficiency in multi-item auctions with decreased marginal utility. The assumption that buyers have decreasing marginal utility is quite common in economics. We propose the Generalized APG Mechanism, which is strategy-proof and individually rational, and prove that it achieves the largest efficiency measure among all strategy-proof mechanisms, where the efficiency measure is the worst-case ratio of the social surplus achieved by the mechanism to the optimal social surplus. 

The following is the main difference between the Generalized Information Diffusion Mechanism (or the Information Diffusion Mechanism) and our approach. An important concept for the two Information Diffusion Mechanisms is a critical diffusion node that shares the information with a buyer who can only get it through the node. Unlike these two mechanisms, an important concept for our mechanisms is the Aligned Path Graph (APG), which is comprised of social networks and prioritizes each buyer. This graph simplifies our mechanisms. 

The following papers are related to ours. \cite{edelman2007internet} investigate generalized second price auctions for the market of internet advertisements. Although its mechanism greatly benefits companies like Google and has some interesting properties, truth-telling is not an equilibrium in it. We focus on truth-telling. \cite{bergemann2002information} study a mechanism design setting where each agent can acquire information about the state of the world before she joins. In contrast, we study a mechanism design setting where each agent can spread the information about an auction's existence. \cite{kempe2003maximizing} analyze a social network where ideas and influence propagate. Their main question is how to choose a set of targeted individuals. Although we have a similar issue concerning how to spread information, we examine auctions.

The structure of this paper is as follows. Section 2 describes a model of auctions and gives definitions. In Section 3, we analyze a single-item auction via a social network. In Section 4, we analyze the multi-item auction, indicate the counter-example, and describe the Generalized APG mechanism. In Section 5, we conclude. 

\section{Model}

Consider a set of buyers, $N=\{1, \ldots, n\}$, and the seller $s$ who wants to sell $k$ identical items. For each $i \in N\cup\{s\}$, s/he can tell the information to the neighbors $R_i \subseteq N \setminus \{i\}$, and each buyer has their valuation profile $v_i = (v_{i,1}, v_{i,2}, \ldots, v_{i,k})$, where $v_{i,1} \geq v_{i,2} \geq \cdots \geq v_{i,k} \geq 0$ on the items. Namely, the buyers have decreasing marginal utility on the items. Define the type of buyer $i$ as $\theta_i = (v_i, R_i)$. We assume that all buyers can declare their type $\theta'_i = (v'_i, R'_i)$, s.t. $R'_i \subseteq R_i$, i.e., we assume buyer $i$ can choose not to send the information to some of her neighbors, but she cannot pretend to send the information to non-neighbors. This is because we assume that each buyer can tell the information to the buyers whom they know, and if they know buyers, they can tell the information to the buyers. We assume that buyer $i$ gets the information if and only if there is a sequence from $s$ to $i$, $(i_0, i_1, \ldots, i_l)$, where $i_{l'+1} \in R'_{i_{l'}}$ for any $l'=0, \ldots,l-1$, $i_0 = s$, and $i_l = i$. Let $\Theta_i$ denote a set of all possible types $i$ can declare, let $\boldsymbol{\Theta}$ denote $\Pi_{i\in N} \Theta_i$, and let $\boldsymbol{\Theta}_{-i}$ denote $\Pi_{j\in N \setminus \{ i \}} \Theta_j$. Define $G$ as the directed graph $(N \cup \{ s \}, E)$, where $N \cup \{ s \}$ is a set of nodes, and $E = \{ (i,j) \in N \cup \{ s \} \times N \mid j \in R'_i \}$ is a set of edges. Let $\boldsymbol{\theta'} = (\theta'_1, \ldots, \theta'_n)$ denote the profile of all buyers' declared types. We use standard notations: $\boldsymbol{\theta'}_{-i}$ is a profile of the buyers' declared types except for $i$, $(\theta'_i, \boldsymbol{\theta'}_{-i})$ is a profile of the buyers' declared types where $i$ declares $\theta'_i$ and the other buyers declare $\boldsymbol{\theta'}_{-i}$. A mechanism $\cal M$ is a pair of allocation function $f_{i}: \boldsymbol{\Theta} \to \mathbb{Z}_{\geq 0}$ and payment function $p_{i,l}: \boldsymbol{\Theta} \to \mathbb{R}$ for all $i \in N$ and all $l=0, \ldots, k$. $f_{i}$ is the number of the items buyer $i$ gets, and $p_{i,l}$ is a payment of the $l$-th item for buyer $i$. When a buyer $i$ gets $k'$ items, the utility of the buyer is $\sum_{l=1}^{k'} (v_{i,l} - p_{i,l})$, and when the buyer cannot get the item, $i$'s utility is $p_{i,0}$. Namely, we assume a standard quasi-linear utility model, and we write the utility of buyer $i$ as $u^{\cal M}_i (\theta_i, \boldsymbol{\theta'}_{-i})$ when $i$'s type is $\theta_i$ and the other buyers declare $\boldsymbol{\theta'}_{-i}$ under a mechanism $\cal M$. We focus on feasible mechanisms: for any $\boldsymbol{\theta'} \in \boldsymbol{\Theta}$, 
\begin{itemize}
\item $\sum_{i \in N} f_{i}(\boldsymbol{\theta'}) \leq k$ 
\item $f_i(\boldsymbol{\theta'}) = p_{i,l}(\boldsymbol{\theta'}) = 0$ for any $i \in N$ and any $l=0,\ldots,k$ if buyer $i$ does not get the information. 
\end{itemize}
The second item means that the buyers who do not get the information cannot participate in the auction. The convention is that these buyers declare a type. 

First of all, we will define strategy-proofness of a mechanism: 

\begin{defi}
A mechanism is strategy-proof if and only if $u^{\cal M}_i((v_i, R_i), \boldsymbol{\theta'}_{-i}) \geq u^{\cal M}_i((v'_i, R'_i), \boldsymbol{\theta'}_{-i})$ for any $(v_i, R_i), (v'_i, R'_i) \in \Theta_i$ s.t. $R'_i \subseteq R_i$ and any $\boldsymbol{\theta'}_{-i} \in \boldsymbol{\Theta}_{-i}$. 
\end{defi}
This strategy-proofness means that not only telling a lie about their valuation but also stopping the information to their neighbors does not benefit for all buyers. Secondly, we will define individual rationality of a mechanism: 

\begin{defi}
A mechanism is individually rational if and only if $u^{\cal M}_i(\theta_i, \boldsymbol{\theta'}_{-i}) \geq 0$ for any $\theta_i \in \Theta_i$ and any $\boldsymbol{\theta'}_{-i} \in \boldsymbol{\Theta}_{-i}$. 
\end{defi}
This means that the utility of each buyer is not smaller than 0 for any declaration of the other buyer when the buyer tells the truth. Thirdly, we will define a efficiency measure: 

\begin{defi}
A mechanism has a competitive ratio of $\alpha$ for efficiency (in short, is $\alpha$-efficient) if and only if 
\begin{equation*}
\alpha = \min_{\boldsymbol{\theta'} \in \boldsymbol{\Theta}} \frac{\sum_{i \in N} \sum_{l=1}^{f_{i}(\boldsymbol{\theta'})} v_{i,l}}{\max_{f' \in \cal{F}} \sum_{i \in N} \sum_{l=1}^{f'_{i}(\boldsymbol{\theta'})} v_{i,l}}, 
\end{equation*}
where $\cal{F}$ is a set of all feasible allocation rules. 
\end{defi}
This is a standard measurement to what extent the mechanism achieves efficiency, as \cite{procaccia2013approximate} propose. This efficiency measure is the worst-case ratio of the social surplus the mechanism achieves to the optimal social surplus. $\alpha$ is in the interval $[0,1]$ by definition, and the mechanism which is 0-efficient is the worst from the viewpoint of efficiency. The mechanism which is 1-efficient is the best, and in this case, we call that the mechanism is efficient. 

Finally, we will inductively define a path graph $P$, where one of the end points is the seller, and the next buyer of a buyer (or the seller) is the one with the youngest number among the remaining buyers who have the shortest distance from the seller under $G$. This is the Aligned Path Graph (APG). Let $P_i^{close}$ be a set of buyers closer than $i$ to the seller, let $P_i^{far}$ be a set of buyers farther away than $i$, and let $l^P_i$ be the distance between the seller and $i$ under $P$. Figure 1 represents how to make an APG. The number in each circle is assigned to each buyer. The top graph represents $G$. The closest buyers to the seller in $G$ are buyers 1 and 6. Then, buyer 1 is the closest in APG, and buyer 6 is the second closest. The second closest buyers to the seller in $G$ are buyers 3 and 5. Then, buyer 3 is the third closest in APG, and buyer 6 is the fourth closest. The third closest buyers to the seller in $G$ are buyers 2 and 4. Then, buyer 2 is the fifth closest in APG, and buyer 4 is the sixth closest. 
\begin{figure}[]
   \centering
    \includegraphics[width=7cm]{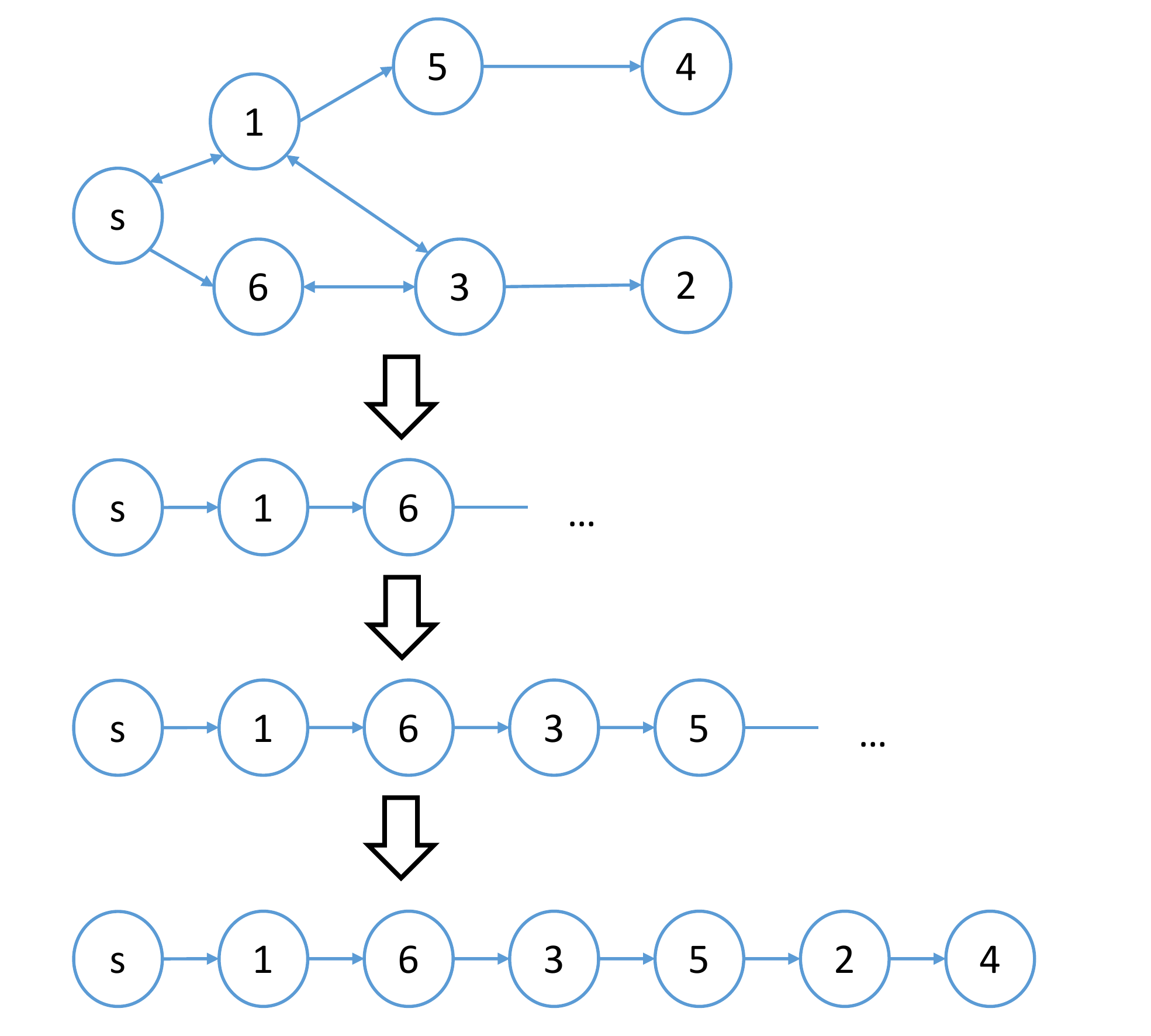}
     \caption{Aligned Path Graph (APG)}
\end{figure}

\section{Mechanisms for single-item auctions}

In this section, we will consider single-item auctions. We write $p_i$ as $p_{i,1}$ when buyer $i$ gets an item and $p_{i,0}$ when buyer $i$ gets no item in this section. We assume the largest valuation is bounded, i.e., for any $\theta_i = (v_i, R_i) \in \Theta_i$, $v_i \leq v^*$ holds in this section. Before defining our mechanism, we will mention two mechanisms, Information Diffusion Mechanism (IDM) and the VCG Mechanism (VCG). In auctions via the social network, IDM and the VCG are proposed in \cite{li2017mechanism}; and prove that both are strategy-proof and individually rational. In addition, they prove that the seller's revenue can be negative in the VCG and propose IDM as a mechanism in which the seller's revenue is always non-negative. 

We will define $\beta$-weakly budget balancedness ($\beta$-WBB), which is the ratio of the worst revenue in the VCG to the worst revenue in a mechanism. 

\begin{defi}
A mechanism is $\beta$-WBB if and only if 
\begin{equation*}
\beta = 1 + \min_{\boldsymbol{\theta'} \in \boldsymbol{\Theta}} \frac{\sum_{i \in N} p_i(\boldsymbol{\theta'})}{(n-1) v^*}. 
\end{equation*}
\end{defi}
Note that for the VCG, $\beta=0$, i.e., in the worst-case, each buyer except the winner gets side-payment that is $v^*$, while the winner pays nothing. If $\beta \geq 1$, the revenue in a mechanism is not negative. If $\beta \geq 1$, we just say that the mechanism is WBB. As we state, IDM is WBB, and the VCG is $0$-WBB. In fact, efficiency is inconsistent with WBB under some basic criteria, as the following proposition states, and \citet{myerson1983efficient} show in the context of bilateral trading. Let $v^*_{S}$ be the largest declared valuation among $S \subseteq N$. 
\begin{prop}
There is no mechanism which is strategy-proof, individually rational, efficient, and weakly budget balanced. 
\end{prop}

\begin{proof}
See Appendix A.1.
\end{proof}

This proposition gives the first axiomatic analysis for auctions via a social network, and this result complements the results by \cite{li2017mechanism}. 

Next; we will propose a class of mechanisms which operate a trade-off between efficiency and weakly budget balancedness in the sense that it is $\alpha$-efficient and $(1-\alpha)$-WBB, where $\alpha \in (0,1)$. We call this class the $\alpha$-APG mechanisms. Define the mechanism $(\{ f_i \}_{i \in N}, \{ p_i \}_{i \in N})$ as follows: 
\begin{equation*}
f_i(\boldsymbol{\theta'}) = \begin{cases}
1 \text{ if } i = \argmin_{j \in M} l^P_j, \\
0 \text{ otherwise, } 
\end{cases}
\end{equation*}
and  
\begin{equation*}
p_i (\boldsymbol{\theta'}) = \begin{cases}
v^*_{P_i^{close}} \text{ if } i = w \text{ and } v^*_{P_i^{close}} < \alpha v^*_{P_i^{far}}, & \text{(Group 1)} \\
- \alpha v^*_{N} + v^*_{P_i^{close}} \text{ if } i \in P_w^{close}, & \text{(Group 2)} \\
\frac{v^*_{P_i^{close}}}{\alpha} \text{ if } i = w \text{ and } v^*_{P_i^{close}} \geq \alpha v^*_{P_i^{far}}, & \text{(Group 3)} \\ 
0 \text{ if } i \in P_w^{far}, & \text{(Group 4)}
\end{cases}
\end{equation*}
where $M = \{ i \in N \mid v'_i \geq \alpha v^*_{N} \}$, and $w$ is the winner. The winner of the mechanism is the closest buyer to the seller in APG whose valuation is not $\alpha$ times smaller than the maximum valuation. The payment is different for a classification of buyers because the threshold of a buyer's declared valuation to get the item is different. As we will show, the sum of the payment when a buyer gets the item and the payment when the buyer does not equals the threshold, which is implied in \cite{myerson1981optimal}. 

We will examine properties of the mechanism. 
\begin{theo}
The $\alpha$-APG mechanisms is $\alpha$-efficient, $(1-\alpha)$-WBB, individually rational, and strategy-proof. 
\end{theo}
\begin{proof}
See Appendix A.2. 
\end{proof}

By this theorem, we can say that more efficient mechanisms have less WBB in our class of the mechanisms. 

\section{Mechanism for multi-item auctions}

In this section, we will consider multi-item auctions. First
of all, we will show a counter-example to strategy-proofness
in Generalized IDM proposed by \cite{zhao2018selling}, which
is tailored to the multi-item auctions when all buyers need only one item.
The summary of the GIDM is as follows.
First of all, the seller sends items to its direct descendants,
where the number of items each descendant receives is equal to
the number of top $k$ buyers in the sub-tree.
When buyer $i$, who is not within top $k$ winners, receives an item/items, 
she can take away an item from her descendant, if she is willing to
pay a certain price. 
The price is determined such that each buyer has
no incentive to misreport her valuation.
More specifically, the price is equal to the decreased amount of
social surplus caused by allocating an item to the buyer.
Here, the social surplus is calculated such that the allocation of
a buyer who takes away an item, as well as 
the allocation of each of top $k$ buyers, who is not taken away an
item, are fixed.\footnote{%
  We found that GIDM described in \cite{zhao2018selling}
  is not strategy-proof. 
  We contacted one of the authors. He admits there is some
  description errors and suggests
  a revision. Our description here is based on his revision.
  More specifically, the constraints in both the allocation rule and the payment rule must be revised. One of the constraints is $\forall {j \in N_i^{out}}, \pi_j(\boldsymbol{\theta'}) = 0$. Replace the constraint to the following: $\forall {j \in N^{opt} \setminus \{ N_i^{out} \cup D_i} \}, \pi_j(\boldsymbol{\theta'}) = 1$. 
  Here, we show the fact that GIDM is not strategy-proof even after
  this revision.}. 

Consider the graph in Figure 2. The seller $s$ is in the top-left
corner, and the number above or beside each circle is the true
valuation of each buyer. We assume that there are 4 items.
Let us examine the incentive for
buyer $d$.

First, let us assume all buyers act truthfully. 
Then, 4 items are sent to buyer $a$ (since $c$, $d$, $e$, and $g$
are top-4 buyers). Then, buyer $a$ takes away
an item from buyer $c$ (whose valuation is smallest).
Next, buyer $b$ receives remaining 3 items and
takes away an item from buyer $e$, and
sends one item to buyer $c$ (as well as to buyer $g$). 
Then, buyer $c$ takes away one item from $d$.
Here, the price of $c$ is equal to $3$. This is because
assuming the assignments of buyers $a$ and $b$ (who take away items) are fixed,
the best way for allocating remaining two items is
allocating them to $c$ and $d$. 
If we do not allocate an item to $c$, the best way to allocating
remaining two items is allocating them to $f$ and $g$. Thus,
the decreased amount of social surplus except $c$ is $3$.
Thus, buyer $d$ cannot obtain an item and her utility is $0$. 

Next, let us assume buyer $d$ stops spreading the information. 
Then, 3 items are sent to buyer $a$ (since $c$, $d$, and $g$
are top-4 buyers), while one item is sent to buyer $f$.
Then, buyer $a$ takes away
an item from buyer $c$. 
Next, buyer $b$ receives remaining 2 items and
takes away an item from buyer $g$, and
sends one item to buyer $c$.
However, in this case, buyer $c$ cannot take away an item,
since her price becomes $6$. 
This is because
assuming the assignments of buyers $a$, $b$, and $f$ are fixed, 
the best way for allocating remaining one item is 
allocating it to $g$. If the item is allocated to $c$ instead,
the social surplus except $c$ decreases by $6$. 
Thus, buyer $c$ does not take away an item and sends it to buyer $d$.
Buyer $d$ obtains an item and pays $6$. Thus, her utility becomes
positive by stopping the information. This violates
strategy-proofness. 

A failure of the proof is on the bottom
of the right column in page 7, where they simply assert 
that stopping the information is never beneficial.
Superficially, this sounds plausible since if a buyer
stops spreading the information, the number of items sent toward
her can never increase. However, our counter-example shows that 
doing so is actually beneficial. 

\begin{figure}[t]
    \centering
    \includegraphics[height=2.5cm]{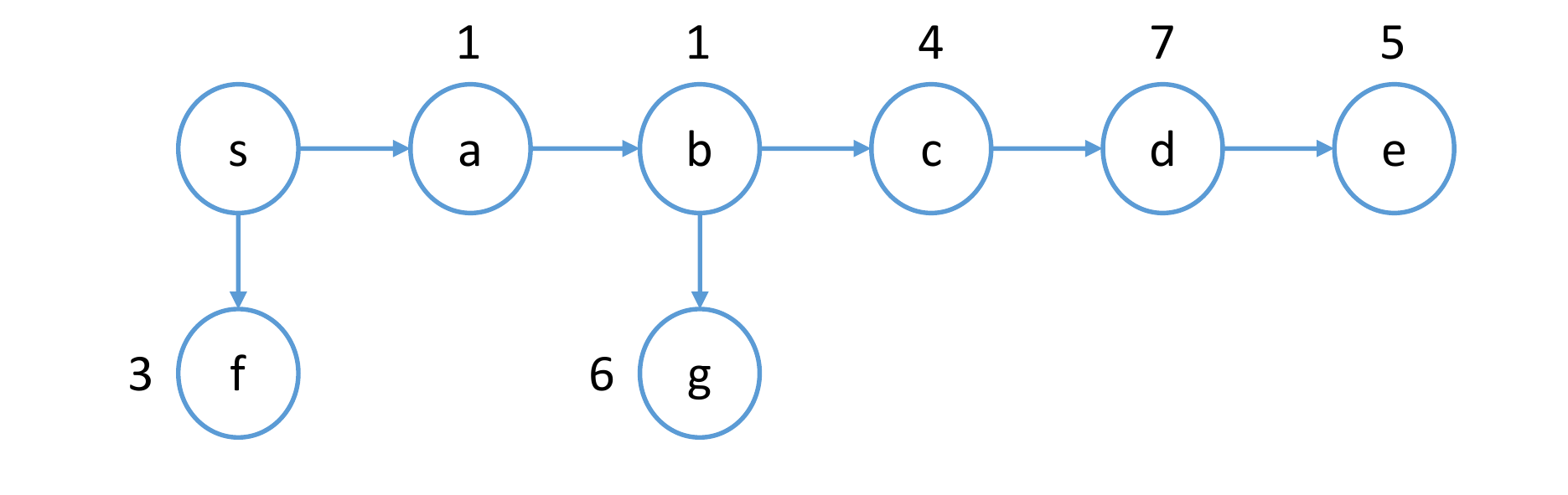}
    \caption{Example: GIDM is not strategy-proof}
    \label{fig:four-items}
\end{figure}

The above counter-example reveals that there is no existing strategy-proof mechanism which is tailored to multi-item auctions when all buyers need only one item. Now, we will analyze wider range of multi-item auctions. We assume that buyers may need multiple items and have decreasing marginal utility, as we defined in Section 2. First of all, we will show an impossibility theorem. Let $n_k = \lfloor \sqrt k \rfloor$ and let $v^*(S,k')$ be the $k'$-th largest declared valuation among $S \subseteq N$. 

\begin{theo}
There is no mechanism which is strategy-proof and $\left( \frac{n_k}{k}+\epsilon \right)$-efficient for any $\epsilon > 0$. 
\end{theo}
\begin{proof}
See Appendix A.3. 
\end{proof}

Next; we will propose a mechanism that is $\frac{n_k}{k}$-efficient. Let $v_{i,sum} = \sum_{l=1}^{n_k} v_{i,l}$ and $v'_{i,sum} = \sum_{l=1}^{n_k} v'_{i,l}$. We call these $n_k$-sum valuations. The winners are the buyers who have  the $n_k$-sum valuation which is not smaller than the $n_k$-th $n_k$-sum valuation. When there is a tie, younger buyers are prioritized. All the winners get just $n_k$ items, and the remaining items are discarded. If a buyer $i$ gets items, the buyer pays $v_{sum}^* (P_i^{close}, n_k)$, where $v_{sum}^* (S, k')$ is the $k'$-th $n_k$-sum valuation in $S \subseteq N$. If buyer $i$ does not get any items, the buyer gets $v_{sum}^* (N, n_k) - v_{sum}^* (P_i^{close}, n_k)$. We call this the generalized APG mechanism (GAPG mechanism). 

Now, we will address properties of the mechanism. 

\begin{theo}
The GAPG mechanism is strategy-proof and $\frac{n_k}{k}$-efficient. 
\end{theo}
\begin{proof}
See Appendix A.4. 
\end{proof}

Theorems 2 and 3 completely answer the question on the approximability of strategy-proof mechanisms with decreasing marginal utility, i.e., Theorem 2 provides an upper bound of $\alpha$, and Theorem 3 shows the bound is tight by providing a strategy-proof mechanism that exactly achieves the upper bound. Additionally, this mechanism is individually rational but neither WBB nor non-wasteful, that is, there is an item no buyer gets. 

On the other hand, when all buyers only require a single item, there exists a efficient mechanism which is almost the same as GAPG. The payment rule of the mechanism is completely the same as GAPG when $n_k = 1$, and the winners are just the top $k$ buyers. We can easily show that the mechanism is efficient and strategy-proof by the similar way to the proof of Theorem 3. 

\section{Concluding Remarks}

We studied the extent to which a mechanism achieves efficiency in auctions via a social network. In single-item auctions, there are two existing mechanisms, the VCG and IDM, which \cite{li2017mechanism} analyzed to show that the VCG is efficient but not WBB, and IDM is WBB but not efficient. We propose the $\alpha$-APG mechanisms, a class of mechanisms which operate a trade-off between efficiency and weakly budget balancedness. When this mechanism achieves a more efficient allocation, it achieves less WBB. 

Second, we indicated a counter-example to strategy-proofness in the GIDM proposed by \cite{zhao2018selling}, which is tailored to multi-item auctions when all buyers need only one item. This increases the importance of our next contribution. Third, we analyzed multi-item auctions with decreasing marginal utility and proposed the generalized APG mechanism, which is $\frac{n_k}{k}$-efficient, strategy-proof, and individually rational. We showed that no mechanism exists that is strategy-proof and has the larger efficiency measure than the GAPG mechanism. In addition, this mechanism is not weakly budget balanced. In multi-item auctions with decreasing marginal utility, since no mechanisms have been proposed that are weakly budget balanced and strategy-proof, future works will find such a mechanism. Moreover, the GAPG mechanism is not non-wasteful, so another important question is to what extent a mechanism can achieve efficiency under strategy-proofness and non-wastefulness. 

\bibliography{reference}
\appendix
\section{Appendix}

\subsection{The proof of Proposition 1} 
Suppose for a contradiction that there exists such a mechanism. Fix $i \in N$ and $\boldsymbol{\theta'}_{-i}$ arbitrarily. First of all, assume that buyer $i$ is the winner when $i$ tells the truth and spreads the information. By efficiency, $v_i \geq v^*_{N_{-i}}$, where $N_{-i}$ is a set of all buyers without $i$. If buyer $i$ tells the truth and spreads the information, the utility of $i$ is $v_i-p_i(\boldsymbol{\theta'}_{-i}, (v_i,R_i))$. If $i$ lies and reports $v'_i < v^*_{N_{-i}}$, $i$'s utility is 
\begin{align*}
& - p_i(\boldsymbol{\theta'}_{-i}, (v'_i,R_i)) \leq v_i - p_i(\boldsymbol{\theta'}_{-i}, (v_i,R_i)) \\
\Leftrightarrow \ & - p_i(\boldsymbol{\theta'}_{-i}, (v'_i,R_i)) + p_i(\boldsymbol{\theta'}_{-i}, (v_i,R_i)) \leq v_i.  
\end{align*}
Since the inequality holds for any $v_i \geq v^*_{N_{-i}}$, 
\begin{equation*}
- p_i(\boldsymbol{\theta'}_{-i}, (v'_i,R_i)) + p_i(\boldsymbol{\theta'}_{-i}, (v_i,R_i)) \leq v^*_{N_{-i}}
\end{equation*}
because $p_i$ is constant except for $v'_i = v^*_{N_{-i}}$. Similarly, since $v_i \geq p_i(\boldsymbol{\theta'}_{-i}, (v_i,R_i))$ by individual rationality, 
\begin{equation*}
v^*_{N_{-i}} \geq p_i(\boldsymbol{\theta'}_{-i}, (v_i,R_i)). 
\end{equation*}

Secondly, assume that buyer $i$ is not the winner when $i$ tells the truth and spreads the information. Then, $v_i \leq v^*_{N_{-i}}$. If buyer $i$ tells the truth and spreads the information, the utility of $i$ is $-p_i(\boldsymbol{\theta'}, (v_i,R_i))$. If $i$ lies and reports $v'_i > v^*_{N_{-i}}$, the utility is 
\begin{align*}
& v_i - p_i(\boldsymbol{\theta'}_{-i}, (v_i,R_i)) \leq - p_i(\boldsymbol{\theta'}_{-i}, (v'_i,R_i)) \\
\Leftrightarrow \ & v_i \leq - p_i(\boldsymbol{\theta'}_{-i}, (v'_i,R_i)) + p_i(\boldsymbol{\theta'}_{-i}, (v_i,R_i)).  
\end{align*}
Since the inequality holds for any $v_i \leq v^*_{N_{-i}}$, 
\begin{equation*}
v^*_{N_{-i}} \leq - p_i(\boldsymbol{\theta'}_{-i}, (v'_i,R_i)) + p_i(\boldsymbol{\theta'}_{-i}, (v_i,R_i)). 
\end{equation*}
Then, 
\begin{equation*}
v^*_{N_{-i}} = - p_i(\boldsymbol{\theta'}_{-i}, (v'_i,R_i)) + p_i(\boldsymbol{\theta'}_{-i}, (v''_i,R_i)) 
\end{equation*}
for any $v'_i < v^*_{N_{-i}}$ and any $v''_i > v^*_{N_{-i}}$. 

Suppose an additional assumption that $v^*_{N_{-i}} > 0$, $v^*_{N_{-i}} > v_i > 0$, and $v'_j = 0$ for any other buyer $j$. Moreover, assume that the buyer whose valuation is $v^*_{N_{-i}}$ cannot get the information if $i$ does not tell the information. In this case, when $i$ stops spreading the information, the utility is 
\begin{equation*}
v_i - 0 \leq - p_i(\boldsymbol{\theta'}_{-i}, (v_i,R_i)). 
\end{equation*}
Then, 
\begin{equation*}
v^*_{N_{-i}} - 0 \leq - p_i(\boldsymbol{\theta'}_{-i}, (v_i,R_i)). 
\end{equation*}
By the above discussion, if $i$ tells the truth and spreads the information, the buyer whose valuation is $v^*_{N_{-i}}$ pays at most $v_i$. As a result, the seller's revenue is at most $v_i - v^*_{N_{-i}} < 0$ by individual rationality. By weakly budget balancedness, this is a contradiction. 

\subsection{The proof of Theorem 1}
Fix  $i \in N$ and $\boldsymbol{\theta'}_{-i} \in \boldsymbol{\Theta'}_{-i}$ arbitrarily. We will show that for any $\boldsymbol{\theta'} \in \boldsymbol{\Theta'}$, the utility of $i$ is maximized when $i$ reports her type truthfully to prove strategy-proofness. First of all, assume that buyer $i$ is classified into Group 1 when $i$ reports her type truthfully. Then, $v_i \geq \alpha v^*_N$ and $v^*_{P_i^{close}} < \alpha v^*_N$. If $i$ reports her type truthfully, the utility of $i$ is $v_i - v^*_{P_i^{close}}$ which does not depend on $\theta'_i$. If $i$ misreports her type and is classified into Group 2, the utility of $i$ is 
\begin{equation*}
\alpha v^*_N - v^*_{P_i^{close}} \leq v_i - v^*_{P_i^{close}}
\end{equation*}
by $v_i \geq \alpha v^*_N$. If $i$ misreports her type and is classified into Group 3, the utility of $i$ is 
\begin{equation*}
v_i - \frac{v^*_{P_i^{close}}}{\alpha} \leq v_i - v^*_{P_i^{close}}
\end{equation*}
by $0 < \alpha < 1$. If $i$ misreports her type and is classified into Group 4, the utility of $i$ is $0 \leq v_i - v^*_{P_i^{close}}$ because $v^*_{P_i^{close}} < \alpha v^*_N \leq v_i$. Therefore, $i$ who is classified into Group 1 when $i$ tells the truth has no incentives to misreport her type. 

Secondly, assume $i$ is classified into Group 2 when $i$ reports her type truthfully. Then, $v_i < \alpha v^*_N$ and $v^*_{P_i^{close}} < \alpha v^*_N$. If $i$ reports her type truthfully, the utility of $i$ is $\alpha v^*_N - v^*_{P_i^{close}}$, which does not depend on $v'_i$ because $i$ does not report the maximum value. In addition, if $i$ diffuses the information to $R'_i \subsetneq R_i$, the utility does not increase because $v^*_N$ does not increase. Then, if $i$ misreports her type and is classified into Group 2, the utility does not increase. If $i$ misreports her type and is classified into Group 1, the utility of $i$ is 
\begin{equation*}
v_i - v^*_{P_i^{close}} < \alpha v^*_N - v^*_{P_i^{close}} 
\end{equation*}
by $v_i < \alpha v^*_N$. If $i$ misreports her type and is classified into Group 3, the utility of $i$ is 
\begin{equation*}
v_i - \frac{v^*_{P_i^{close}}}{\alpha} \leq v_i - v^*_{P_i^{close}} < \alpha v^*_N - v^*_{P_i^{close}}. 
\end{equation*}
If $i$ misreports her type and is classified into Group 4, the utility of $i$ is 
\begin{equation*}
0 < \alpha v^*_N - v^*_{P_i^{close}} 
\end{equation*}
because $v^*_{P_i^{close}} < \alpha v^*_N$. Therefore, $i$ who is classified into Group 2 when $i$ reports her type truthfully, has no incentives to misreport her type. 

Thirdly, assume $i$ is classified into Group 3 when $i$ reports her type truthfully. Then, $v_i \geq \alpha v^*_N$ and $v^*_{P_i^{close}} < \alpha v^*_N$. In addition, $i = m$ because $i = w$ and $v^*_{P_i^{close}} \geq \alpha v^*_{P_i^{far}}$. Notice that $i$ is classified into either Group 3 or Group 4 regardless what $i$ reports. This is because if $i$ does not spread the information, $v^*_{P_i^{close}}$ is unchanged, and $v^*_{P_i^{far}}$ can decrease. If $i$ tells the truth, the utility of $i$ is $v_i - \frac{v^*_{P_i^{close}}}{\alpha}$ which does not depend on $\theta'_i$. If $i$ misreports her type and is classified into Group 4, the utility of $i$ is 
\begin{equation*}
0 < v_i - \frac{v^*_{P_i^{close}}}{\alpha}
\end{equation*}
by $v^*_{P_i^{close}} < \alpha v^*_N = \alpha v_i$. Therefore, $i$ who is classified into Group 3 when $i$ reports her type truthfully, has no incentives to misreport her type. 

Finally, assume $i$ is classified into Group 4 when $i$ reports her type truthfully. Then, $v^*_{P_i^{close}} \geq \alpha v^*_N$. Since $i \neq w$, $\alpha v_i \ ( \ \leq v'_w) \leq v^*_{P_i^{close}}$. $i$ is also classified into either Group 3 or Group 4 even if $i$ misreports her type. If $i$ tells the truth, the utility of $i$ is $0$ which does not depend on $\theta'_i$. If $i$ misreports her type and is classified into Group 3, the utility of $i$ is 
\begin{equation*}
v_i - \frac{v^*_{P_i^{close}}}{\alpha} \leq 0
\end{equation*}
by $\alpha v_i \leq v^*_{P_i^{close}}$. Therefore, $i$ who is classified into Group 4 when $i$ reports her type truthfully, has no incentives to misreport her type. As a result, the $\alpha$-APG mechanisms is strategy-proof. 

Next, we will prove individual rationality. If $i$ is classified into Group 1, Group 2, or Group 4 when $i$ tells the truth, the mechanism is obviously individually rational. Consider that $i$ is classified into Group 3 when $i$ tells the truth. Then, $v^*_{P_i^{close}} < \alpha v_i$ by the above discussion about Group 3. Thus, 
\begin{equation*}
v_i - \frac{v^*_{P_i^{close}}}{\alpha} = \frac{1}{\alpha} \left\{ \alpha v_i - v^*_{P_i^{close}} \right\} > 0. 
\end{equation*}
As a result, the $\alpha$-APG mechanisms is individually rational. 

By the definition of the mechanism, it is $\alpha$-efficient. Then; we will show that it is $(1-\alpha)$-WBB and complete the proof. By its definition, $\sum_{i \in N} p_i(\boldsymbol{\theta'})$ is minimized when $v^* = v^*_N$, and the graph $G(\boldsymbol{\theta})$ is a path graph. In this case, 
\begin{equation*}
1 + \min_{\boldsymbol{\theta'} \in \boldsymbol{\Theta'}} \frac{\sum_{i \in N} p_i(\boldsymbol{\theta'})}{(n-1) v^*} = 1 - \frac{(n-1) \alpha v^*_N}{(n-1) v^*} = 1 - \alpha. 
\end{equation*}
As a result, the $\alpha$-APG mechanisms is $(1-\alpha)$-WBB. 

\subsection{The proof of Theorem 2}
Assume for a contradiction that there exists such a mechanism. First of all, we will show that there is no buyer who gets $k'$ items, where $n_k \geq k' > 0$. Suppose for a contradiction that there exists such a buyer $i$ for some $\theta'_i$ and some $\boldsymbol{\theta'}_{-i}$. If $v'_{i,1} = v'_{i,2} = \cdots =v'_{i,k} > \frac{k-n_k}{k \epsilon} v^* (N \setminus \{ i \},1)$, buyer $i$ must get more than $n_k$ items by $\left( \frac{n_k}{k}+\epsilon \right)$-efficiency. Then, assume that buyer $i$ gets $k'$ items when $i$ reports $\theta'_i$ and gets $k'' > n_k$ items when $i$ reports $\theta''_i$. 

We, first, assume that $i$'s type is $\theta'_i$. When buyer $i$ tells the truth, the utility is $\sum_{l=1}^{k'} \{ v_{i,l} - p_{i,l} \}$. When buyer $i$ lies and gets $k''$ items, the utility is 
\begin{equation*}
\sum_{l=1}^{k''} \{ v_{i,l} - p_{i,l} \} \leq \sum_{l=1}^{k'} \{ v_{i,l} - p_{i,l} \}
\end{equation*}
by strategy-proofness. Then, 
\begin{equation*}
\sum_{l=k'}^{k''} v_{i,l} \leq \sum_{l=k'}^{k''} p_{i,l}. 
\end{equation*}
Then, there must be an upper bound of $\sum_{l=k'}^{k''} v_{i,l}$ under the constraint that buyer $i$ gets $k'$ items by telling the truth because the mechanism is strategy-proof. Let $ub$ be the upper bound. We get $ub \leq \sum_{l=k'}^{k''} p_{i,l}$ by strategy-proofness. 

Next, assume that $i$'s type is $\theta''_i$. When buyer $i$ tells the truth, the utility is $\sum_{l=1}^{k''} \{ v_{i,l} - p_{i,l} \}$. When buyer $i$ lies and gets $k'$ items, the utility is 
\begin{equation*}
\sum_{l=1}^{k'} \{ v_{i,l} - p_{i,l} \} \leq \sum_{l=1}^{k''} \{ v_{i,l} - p_{i,l} \}
\end{equation*}
by strategy-proofness. Then, 
\begin{equation*}
\sum_{l=k'}^{k''} v_{i,l} \geq \sum_{l=k'}^{k''} p_{i,l}. 
\end{equation*}
Let $lb$ be the lower bound of $\sum_{l=k'}^{k''} v_{i,l}$ under the constraint that buyer $i$ gets $k'$ items by telling the truth. Then; $lb \geq \sum_{l=k'}^{k''} p_{i,l}$ holds by strategy-proofness. Therefore, $ub = lb = \sum_{l=k'}^{k''} p_{i,l}$. Therefore, $ub$ does not depend on $i$'s own valuations. Since $\frac{k-n_k}{k \epsilon} v^* (N \setminus \{ i \},1) \geq ub = lb$ holds by $\left( \frac{n_k}{k}+\epsilon \right)$-efficiency, $ub = lb =0$ when $v^* (N \setminus \{ i \},1) = 0$. In addition, $ub$ is not an increasing function of $v^* (N \setminus \{ i \},1)$ because the buyer whose valuation is $v^* (N \setminus \{ i \},1)$ may only get the information through $i$. Therefore, $ub = lb =0$ for any $\boldsymbol{\theta'}_{-i}$, which contradicts the existence of a buyer who gets $k'$ items. 

Then, all buyers who get some items get $n_k+1$ items or more. Then, the number of winners is $n_k$ or fewer. Consider $v_{1,1} = v_{n,1} > 0$ and $v_{1,l} = v_{n,l} = 0$ for any $l=2,\ldots,k$. Then, this mechanism achieves $\left( \frac{n_k}{k} \right)$-efficiency or less. This is a contradiction. 

\subsection{The proof of Theorem 3}
First of all, we will show its strategy-proofness. Fix $i \in N$ and $\boldsymbol{\theta'}_{-i} \in \boldsymbol{\Theta}_{-i}$ arbitrarily. We will show that buyer $i$ has no incentives to lie or to stop spreading the information. We, first, assume that buyer $i$ gets nothing when the buyer tells the truth. Then, $v_{i,sum} \leq v_{sum}^* (N, n_k)$ and buyer $i$ gets $v_{sum}^* (N, n_k) - v_{sum}^* (P_i^{close}, n_k)$ if $i$ tells the truth and spreads the information. Notice that $v_{sum}^* (P_i^{close}, n_k)$ does not depend on $R_i$, and $v_{sum}^* (N, n_k)$ decreases when buyer $i$ does not spread the information. Then, buyer $i$ does not have incentives to stop spreading the information. If buyer $i$ lies and becomes a winner, the utility is 
\begin{equation*}
v_{i,sum} - v_{sum}^* (P_i^{close}, n_k) \leq v_{sum}^* (N, n_k) - v_{sum}^* (P_i^{close}, n_k) 
\end{equation*}
because $v_{i,sum} \leq v_{sum}^* (N, n_k)$. Then, buyer $i$ has no incentives to lie. Secondly, assume that buyer $i$ is a winner. Then, $v_{i,sum} \geq v_{sum}^* (N, n_k)$. When $i$ tells the truth, the utility is 
\begin{equation*}
v_{i,sum} - v_{sum}^* (P_i^{close}, n_k). 
\end{equation*}
Notice that the utility does not depend on both $v'_i$ and $R'_i$. Then, $i$ does not have incentives to stop spreading the information. If buyer $i$ lies and becomes a loser, the utility of buyer $i$ is 
\begin{equation*}
v_{sum}^* (N, n_k) - v_{sum}^* (P_i^{close}, n_k) \leq v_{i,sum} - v_{sum}^* (P_i^{close}, n_k)
\end{equation*}
because $v_{i,sum} \geq v_{sum}^* (N, n_k)$. As a result, the GAPG mechanism is strategy-proof. 

Secondly, we will show that the mechanism is $\frac{n_k}{k}$-efficient. Let $v_{i,sum}^l$ be the $l$-th $n_k$-sum valuation for any $l=1, \ldots, n$. Since $v_{i,sum}^l \geq v^*(N,l)$ for any $l=1,\ldots,n_k$, 
\begin{equation*}
\frac{\sum_{l=1}^{n_k} v_{i,sum}^l}{\sum_{l=1}^k v^*(N,l)} \geq \frac{n_k v_{i,sum}^{n_k}}{n_k v^*(N,n_k) + \sum_{l=n_k+1}^k v^*(N,l)} \geq \frac{n_k}{k}
\end{equation*}
for any $\boldsymbol{\theta'}$. By Theorem 2, the GAPG mechanism is $\frac{n_k}{k}$-efficient. 

\end{document}